\documentclass{article}
\usepackage{amsmath,amsfonts,amssymb}
\usepackage{amsthm}
\usepackage{mathrsfs}
\usepackage{hyperref}

  \newcommand{\oo}{\infty}
  \newcommand{\del}{\partial}
\renewcommand{\d}{\mathrm{d}}
\renewcommand{\dh}{\mathrm{d}_{\mathsf{h}}}
  \newcommand{\dv}{\mathrm{d}_{\mathsf{v}}}

  \newcommand{\sso}{\subset}
  \newcommand{\sse}{\subseteq}

  \newcommand{\B}{\mathcal{B}}
  \newcommand{\E}{\mathcal{E}}
  \newcommand{\EL}{\mathrm{EL}}
\renewcommand{\L}{\mathcal{L}}

  \newcommand{\Secs}{\mathrm{\Gamma}}
  \newcommand{\Forms}{\mathrm{\Omega}}
  \newcommand{\Lie}{\mathscr{L}}

\theoremstyle{plain}
\newtheorem{theorem}{Theorem}
\newtheorem{lemma}[theorem]{Lemma}
\newtheorem{proposition}[theorem]{Proposition}

\theoremstyle{definition}
\newtheorem{definition}{Definition}

\title{Presymplectic current and the inverse problem of the calculus of variations}
\author{Igor Khavkine\\
Institute for Theoretical Physics, Utrecht, Leuvenlaan 4,\\
NL-3584 CE Utrecht, The Netherlands\\
\texttt{i.khavkine@uu.nl}}

\begin{document}
\maketitle
\begin{abstract}
The inverse problem of the calculus of variations asks whether a given
system of partial differential equations (PDEs) admits a variational
formulation. We show that the existence of a presymplectic form in the
variational bicomplex, when horizontally closed on solutions, allows us
to construct a variational formulation for a subsystem of the given PDE.
No constraints on the differential order or number of dependent or
independent variables are assumed. The proof follows a recent
observation of Bridges, Hydon and Lawson and generalizes an older result
of Henneaux from ordinary differential equations (ODEs) to PDEs.
Uniqueness of the variational formulation is also discussed.
\end{abstract}

\section{Introduction}\label{sec:intro}
Many systems of partial differential equations (PDEs for short) that
appear in physics are variational. That is, they are equivalent to
Euler-Lagrange (EL) equations of some local Lagrangian density. The
question of whether some arbitrarily given PDE is variational is known
as the \emph{inverse problem of the calculus of variations}. It has
attracted a significant amount of attention in the
past~\cite{anderson-aspects,saunders}. In physics, a variational
formulation endows the algebra functions on the phase space of a
classical theory with Poisson structure, which in turn determines the
corresponding commutation relations in quantum theory. One important
consequence of the inverse problem is the ability to reconstruct a
variational principle and hence the quantum commutation rules (whether
unique or not) directly from the equations of motion~\cite{wigner,hen82}.

This problem comes in two versions, the harder \emph{multiplier} problem
and the easier \emph{non-multiplier} problem. The idea behind the
solution of the non-multiplier problem for ordinary differential
equations (ODEs) was known already to Helmholtz and the necessary and
sufficient conditions for a positive solution are known as the
\emph{Helmholtz conditions}. The generalization of these conditions to
PDEs is elegantly stated in terms of the variational
bicomplex~\cite{anderson-var,anderson-var2}. The multiplier problem is
much less understood, with significant results obtained for ODEs by
Douglas~\cite{douglas} and for PDEs by Anderson and Duchamp~\cite{ad}.
A criterion certifying a positive solution in the case of second order
ODEs was given later by Henneaux~\cite{hen82}, which was later extended
to higher orders~\cite{juras,at}. This criterion is essentially the
existence of a (pre)symplectic form on the space of dependent variables
that is conserved by the flow of the ODE.

The work of Henneaux proved difficult~\cite{hen84} to generalize to partial
differential equations (PDEs). For one thing, it is not immediately
clear what is the right analog of the symplectic form. Also, Henneaux
represented a second order ODE as a vector field on the tangent bundle
of the configuration space. Finally, the conservation condition was expressed
as annihilation by the Lie derivative of the ODE vector field (with
suitable generalization in the time dependent case). However, we can now
say that Henneaux's argument can in fact be (at least partially)
generalized to the PDE case. The analog of the symplectic form
is the covariant presymplectic current (Section~\ref{sec:symp}). And the
analog of the geometric vector field representation of an ODE is the geometric
representation of a PDE as a submanifold of the jet bundle of the
dependent variables (Section~\ref{sec:jets}).  Within this context, the
corresponding generalization of Henneaux's criterion certifying a
positive solution of the multiplier inverse problem for PDEs was given
only very recently in~\cite{bhl,hydon}. This key observation appeared as a
side remark in that work, with only a sketch of the details and without
placing it within the context of the inverse problem. Unfortunately,
this solution of the inverse problem is only partial, as the Lagrangian
density produced by the proposed procedure (Section~\ref{sec:invprob})
produces a variational system that need not be equivalent to the
original PDE system, though this variational system will admit all
solutions of the original PDE system among its own solutions.

This note aims at highlighting and clarifying the above result by
placing it in the appropriate geometric context. Section~\ref{sec:jets}
introduces some basic background on the geometric formulation of PDEs
using jet bundles and the associated variational bicomplex.
Section~\ref{sec:symp} shows how a conserved presymplectic current
density arises in variational problems. Conversely,
Section~\ref{sec:invprob} shows that the presence of a conserved
presymplectic current density is equivalent to the fact that the
solutions of the given PDE system are also solutions (but perhaps not
the only ones) of a variational PDE system. Section~\ref{sec:uniq}
discusses the non-uniqueness of the Lagrangian constructed in
Section~\ref{sec:invprob}.  Finally, conclusions are in
Section~\ref{sec:concl}.

\section{Jet bundles and PDE systems}\label{sec:jets}
This section briefly defines the basic notions needed for the geometric
formulation of PDE systems in terms of jet bundles, including associated
geometric structures, and fixes some notation. For a detailed discussion
of jet bundles see~\cite{olver,nat-diff}. The variational
bicomplex is discussed in detail in~\cite{anderson-var,anderson-var2}.

Let $F\to M$ be a vector bundle over the base space $M$, an
$n$-dimensional smooth manifold. Since we will be only interested in
local questions, assume that both $F$ and $M$ are topologically trivial
(contractible). The arguments in this paper can be straightforwardly
generalized to smooth bundles with non-trivial fiber and base space
topologies~\cite{sard}. Let $\Secs(F)$ denote the space of smooth
sections of $F$. Denote by $T^*N$ the cotangent bundle for any smooth
manifold $N$ and by $\Lambda^k N = \bigwedge^k_M T^*N$ the bundle of
alternating $k$-forms over $N$. Denote by $\Forms^k(N) = \Secs(\Lambda^k
N)$ the space of differential forms on $N$, by $\d$ the corresponding
de~Rham differential and by $\Forms^*(N)$ the total complex of
differential forms graded by degree.

The $k$-jet $j_x^k\phi$ of a section $\phi\in \Secs(F)$ at $x\in M$ can
be defined as the equivalence class of sections $\psi\in \Secs(F)$ that
have coinciding Taylor polynomials up to and including order $k$ at
point $x$ in any (bundle adapted) local coordinate system. Thus, $k$-jets
are a coordinate invariant way of capturing the derivatives of a section
up to order $k$. Local coordinates $(x^i,u^a)$ on $F$ naturally extend
to local coordinates $(x^i,u^a_I)$, $|I|\le k$ on $J^kF$, where $I$ with
$|I|=l$ is a multi-index $(i_1i_2\cdots i_l)$. The coordinates of the
$k$-jet of $\phi$ at $x$ are
\begin{equation}
	(x^i,u^a,u^a_i,\ldots,u^a_I)(j_x^k\phi)
		= (x^i,\phi^a,\del_i\phi^a, \ldots, \del_I \phi^a)(x), ~~ |I|=k,
\end{equation}
where we have used the shorthand $\del_I = \del_{i_1}\cdots\del_{i_k}$.
These $k$-jets form the space $J_x^kF$, which is a fiber of the vector
bundle $J^kF\to M$ of $k$-jets over $M$. The $0$-jets are identical with
the underlying bundle, $J^0F \cong F$. Jet bundles come with natural
projections $J^kF\to J^lF$, for any $k\ge l$, which simply discard the
information about all derivatives higher than $l$. This projection gives
$J^kF$ the structure of an affine bundle over $J^lF$. The projective
limit
\begin{equation}
	J^\oo F = \varprojlim_k J^kF \to \cdots
		\to J^{2}F \to J^{1}F \to J^{0}F \to M
\end{equation}
is called the $\oo$-jet bundle. A smooth function on $J^\oo F$ is the
pullback of a smooth function on some $J^kF$ for some $k<\oo$. That is,
a smooth function on $J^\oo F$ always depends only on finitely many
components of an $\oo$-jet given to it as argument.

A section $\phi\in \Secs(F)$ naturally gives rise to the section
$j^k\phi\colon M \to J^kF$, with $j^k\phi(x) = j_x^k\phi$, called the
$k$-jet prolongation of $\phi$. Similarly, a (not necessarily linear)
bundle morphism $f\colon J^kF\to E$ (which determines a differential
operator $f\circ j^k \colon \Secs(F) \to \Secs(E)$ of order $k$),
naturally gives rise to a bundle morphism $p^l f\colon J^{k+l}F \to J^l
E$, with $p^lf(j_x^{k+l}\phi) = j_x^l(f\circ j^k(\phi))$, called the
$l$-prolongation of $f$.

In the geometric formulation, a PDE system on $F$ of order $k$ is a
submanifold $\iota\colon \E\sso J^kF$ that satisfies the regularity
conditions of being closed and that $\E\to M$ is a smooth sub-bundle of
$J^kF\to M$.  To connect with the usual notion of a PDE, we note that
there always exists (at least up to possible global topological
obstructions~\cite[\textsection 7]{goldschmidt}) a vector bundle $E\to M$ and a bundle
morphism $f\colon J^kF \to E$ such that a section $\phi\in \Secs(F)$
satisfies the system of differential equations $f\circ j^k(\phi) = 0$
iff the image of $j^k\phi\in \Secs(J^kF)$ is contained in $\E$, that is
$j^k\phi$ is actually a section of $\E\to M$. We call the pair $(f,E)$
an equation form of $\E$, which is in general not unique. The equation
form can always be chosen to be regular, which means that any smooth
function $L$ on $J^kF$ vanishes on $\E$ iff it can be written, in local
coordinates $(x^i,u^a)$ and $(x^i,v^B)$ on $F$ and $E$, as $L(x^i,u_I^a)
= L_B(x^i,u_I^a) f^B(x^i,u_I^a)$, with smooth coefficients $L_B$.
Conversely, any regular equation form $(f,E)$ defines a PDE system
$\E_f$ given by the zero set of $f$. The $l$-prolongation $\iota_l
\colon \E^l\sso J^{k+l}F$ is defined as the PDE system $\E_{p^l f}$
corresponding to the equation form $(p^l f,J^lE)$. The $l$-prolongation
comes with with a natural projection $p_l\colon\E^l\to \E$, which simply
restricts the appropriate projection of jet bundles. Note that $p_l$
need not be surjective if the PDE system has non-trivial integrability
conditions. Below, we deal with $\oo$-prolongations $\iota_\oo\colon
\E^\oo\sso J^\oo F$ and $(p^\oo f,J^\oo E)$. We assume that both are
regular.

The de~Rham differential acting on $\Forms^*(J^\oo F)$ can be naturally
written as the sum
\begin{equation}
	\d = \dh + \dv
\end{equation}
of the, respectively, horizontal and vertical differentials. Each is
nilpotent and they anticommute:
\begin{equation}
\dh^2 = \dv^2 = 0 ,
	\quad
\dh\dv + \dv\dh = 0 .
\end{equation}
If $\d_M$ is the de~Rham differential on $M$, $\phi\in \Secs(F)$ is any
section and $L\in\Forms^*(J^\oo F)$ any differential form, then the
defining property of the horizontal differential is that it is
intertwined with $\d_M$ via the pullback of differential forms along
sections:
\begin{equation}
	(j^\oo\phi)^* \dh L = \d_M (j^\oo\phi)^* L .
\end{equation}
Differential forms on $J^\oo F$ have the following natural subspaces:
the purely horizontal forms $\Forms^{h,0}(F)$ that are generated by
pullbacks of forms from $\Forms^h(M)$ along the natural projection
$J^\oo F\to M$ and the purely vertical forms $\Forms^{0,v}(F)\sso
\Forms^v(J^\oo F)$ that are annihilated by the pullback $(j^\oo \phi)^*$
of any section $\phi\colon M\to F$. Purely horizontal and purely
vertical forms generate $\Forms^*(J^\oo F)$ as a graded algebra. The
subspaces of homogeneous horizontal and vertical degrees are denoted by
$\Forms^{h,v}(F)\sso \Forms^{h+v}(J^\oo F)$. The differentials $\dh$ and
$\dv$ are then, respectively, of purely horizontal degree $1$ and of
purely vertical degree $1$. In local coordinates $(x^i,u^a_I)$,
horizontal forms are generated by $\dh x^i = \d x^i$, while vertical
forms are generated by $\dv u^a_I$. The total bi-differential bi-graded
algebra $(\Forms^{*,*}(F),\dh,\dv)$ is called the variational bicomplex
of the vector bundle $F\to M$. Within this bicomplex, we can define the
horizontal, $H^{*,*}(\dh)$, and vertical, $H^{*,*}(\dv)$, cohomology
groups in the obvious way.

The horizontal and vertical degrees, as well as differentials, survive
restriction to the $\oo$-prolonged PDE system $\iota_\oo\colon \E^\oo
\sso J^\oo F$. Thus, the differential forms $\Forms^*(\E^\oo)$ can also
be given the structure of a bi-differential bi-graded algebra
$(\Forms^{*,*}_\E(F),\dh^\E,\dv^\E)$, where $\Forms^{*,*}_\E(F) \cong
\Omega^*(\E^\oo)$. The cohomology groups $H_\E^{*,*}(\dh)$ are called
the characteristic cohomology groups of the PDE system $\E$.
Characteristic cohomology classes can be identified with important
geometric and algebraic properties of $\E$. For instance, a
representative of an element of $H^{n-1,0}_\E(\dh)$ is a non-trivial
conservation law of $\E$~\cite{bbh,bg,kv}.

\section{Covariant presymplectic current}\label{sec:symp}
With notation following Section~\ref{sec:jets}, a \emph{local action
functional} of order $k$ on $F$ is a function $S[\phi]$ of sections
$\phi\in \Secs(F)$,
\begin{equation}
	S[\phi] = \int_M (j^k\phi)^*\L,
\end{equation}
where the \emph{Lagrangian density} $\L\in \Forms^{n,0}(F)$ is a purely
horizontal $n$-form on $J^k F$. It is local because, given a section
$\phi$ and local coordinates $(x^i,u^a_I)$ on $J^kF$, the pullback at
$x\in M$ can be written as 
\begin{equation}
	((j^k\phi)^*\L)(x) = \L(x^i,\del_I\phi^a(x)) ,
\end{equation}
which depends only on $x$ and on the derivatives of $\phi$ at $x$ up to
order $k$. The integral over $M$ will be considered formal, since all
the necessary properties will be derived from $\L$. Incidentally the
usual variational derivative of variational calculus can be put into
direct correspondence with the vertical differential $\dv$ on this
complex, which is how the name \emph{variational bicomplex} was
established.

In Section~\ref{sec:jets} we introduced the variational
bicomplex $(\Forms^{h,v}(F),\dh,\dv)$ of vertically and horizontally
graded differential forms on $J^\oo F$. A Lagrangian density, being of
top horizontal degree, is then a closed element of $\Forms^{n,0}(F)$.
Usually, Lagrangian densities are considered equivalent if they differ
by a horizontally exact term, for example $\L$ and $\L + \dh\B$, where
$\dh\B$ is often referred to as a \emph{boundary term}. In other words,
we should think of $\L$ not just as an element of $\Forms^{n,0}(F)$, but
rather a representative of a cohomology class $[\L]\in H^{n,0}(\dh)$.

Note that even though $\L$ can be thought of as a form on $J^kF$, below
we will carry out all calculations on $J^\oo F$, with the proviso that
all intermediate formulas could have been projected onto jet bundles of
some finite order, bounded throughout the calculation. When necessary,
we shall make use of a local coordinate system $(x^i,u^a_I)$ on $J^\oo
F$.

Using integration by parts if necessary, we can always write the first
vertical variation of the Lagrangian density as
\begin{equation}\label{eq:dvL}
	\dv \L = \EL_a\wedge\dv{u^a} - \dh\theta.
\end{equation}
All terms proportional to $\dv{u^a_I}$, $|I|>0$, have been absorbed into
$\dh\theta$. In the course of performing the integrations by parts,
the coefficients $\EL_a$ can acquire dependence on jets up to order
$2k$. The form $\theta$ is not uniquely specified, as one can freely
substitute $\theta \to \theta + \dh\sigma$, where $\sigma$ is any
form%
	\footnote{Actually, $\dh\sigma$ could also be replaced by any merely
	\emph{closed} form. However, the cohomology groups $H^{*,v}(F)$ for
	$v>0$ are always trivial~\cite[Ch.5]{anderson-var}, so there is no loss in
	generality.} %
in $\Forms^{n-2,1}(F)$. Thus, the jet order of $\theta$ is not bounded
from above (due to the arbitrariness in $\sigma$). However, as can be
seen from integration by parts in local coordinates, $\theta$ can always
be chosen to depend on jets of order no higher than $2k-1$.  The forms
$\EL_a$ define a bundle morphism
\begin{equation}
	\EL\colon J^{2k} F \to \tilde{F}^* = \Lambda^n M\otimes_M F^*,
\end{equation}
where $F^*\to M$ is the dual vector bundle to $F\to M$. The equation
form $(\EL,\tilde{F}^*)$ determines the so-called \emph{Euler-Lagrange
(EL) PDE system} $\iota\colon \E_\EL\sso J^{2k}F$ associated with the
Lagrangian density $\L$ or equivalently the local action functional
$S[\phi]$.

A PDE system with an equation form given by Euler-Lagrange equations of
a Lagrangian density is said to be \emph{variational}. The form
$\theta\in \Forms^{n-1,1}(F)$ is referred to as the \emph{presymplectic
potential current density}.  Applying the vertical differential to
$\theta$ we obtain the \emph{presymplectic current density} (or the
\emph{presymplectic current
density defined by $\L$} if the extra precision is necessary)
\begin{equation}\label{eq:omega-def}
	\omega = \dv \theta,
\end{equation}
with $\omega\in \Forms^{n-1,2}(F)$. Since the vertical differential does
not increase the jet order, the jet order of $\omega$ is bounded by that
of $\theta$. Note that changing $\L$ by a boundary term $\dh\B$ changes
$\theta$ by the vertically exact term $\dv\B$. Hence, $\omega$ is not
altered by this change. If one changes $\theta$ by the addition of a
horizontally exact term $\dh\sigma$, then the presymplectic current
density also changes by a horizontally exact term, $\omega \to \omega -
\dh\dv\sigma$.

The terminology for $\omega$ comes from classical field theory in
physics. Restricted to solution of $\E_\EL$ and integrated over a closed
codim-$1$ surface in $M$ (most often a Cauchy surface) $\omega$ defines
a presymplectic form on the space of solutions of the EL equations. This
(possibly infinite dimensional) space, known as the phase space, is then
a presymplectic manifold. If the initial value problem on this surface
of integration is well posed, it is even a symplectic manifold. This
method of construction the symplectic form on the phase space of a
classical field theory is known as the \emph{covariant phase space
method}~\cite{zuck,abr,crn-wit,lw,bhs}.

Since $\omega$ is, like $\L$, to be integrated over a boundaryless
submanifold, we only care about its equivalence class $[\omega]\in
\Forms^{n-1,2}(F)/\dh\Forms^{n-2,2}(F)$. Thus, the discussion above has
shown the following
\begin{proposition}
	Given a Lagrangian density $\L\in \Forms^{n,0}(F)$ of order $k$, there
	exist forms $\theta\in \Forms^{n-1,1}(F)$ and $\omega\in
	\Forms^{n-1,2}(F)$, of jet order at most $2k-1$, as well as a section
	$\EL\colon J^{2k}\to \tilde{F}^*$ that satisfy the equations
	\begin{align}
		\dv\L &= \EL_a\wedge\dv{u^a_I} - \dh\theta \\
		\dv\theta &= \omega .
	\end{align}
	Moreover, the equivalence class $[\L]\in H^{n,0}(\dh)$ determines the
	section $\EL$ and the equivalence class $[\omega]\in
	\Forms^{n-1,2}(F)/\dh\Forms^{n-2,2}(F)$ uniquely.
\end{proposition}

The following lemma is an easy consequence of the definition of
$\omega$.
\begin{lemma}
When pulled back along the inclusion $\iota\colon \E_\EL\sse J^{2k} F$,
the image $\iota^*\omega \in \Omega_{\E_\EL}^{n-1,2}(F)$ of the form
$\omega$ Equation~\eqref{eq:omega-def} is both horizontally and
vertically closed:
\begin{align}
	\dh^{\E_\EL} \iota^* \omega &= 0 , \\
	\dv^{\E_\EL} \iota^* \omega &= 0 .
\end{align}
\end{lemma}
\begin{proof}
The horizontal and vertical differentials on $\E_\EL$ are defined by
pullback along $\iota$, that is, $\dh^{\E_\EL} \iota^* = \iota^* \dh$
and $\dv^{\E_\EL} \iota^* = \iota^* \dv$. Since $\omega=\dv\theta$ is
already vertically closed in $\Forms^{n-1,2}(F)$, it is a fortiori
vertically closed in $\Forms_{\E_\EL}^{n-1,2}(F)$. The rest is a
consequence of the nilpotence and anti-commutativity of $\dh$ and $\dv$:
\begin{align}
	0 = \dv^2\L
		&= \dv\EL_a\wedge\dv{u^a} - \dv\dh\theta, \\
	\dh\omega
		&= \dh\dv\theta = -\dv\dh\theta
		= -\dv \EL_a\wedge\dv{u^a} , \\
	\dh^{\E_\EL} \iota^*\omega
		&= \iota^*\dh\omega
		= -\iota^* \dv\EL_a \wedge \dv u^a
		= 0 ,
\end{align}
where the last equality holds because $\EL_a$ and $\dv\EL_a$ generate
the ideal in $\Forms^{*,*}(F)$ annihilated by the pullback $\iota^*$.
\end{proof}

In fact, we will promote the name \emph{presymplectic current density}
to any form satisfying these properties.
\begin{definition}\label{def:presymp}
Given a PDE system $\iota\colon \E \sso J^kF$ we call a form
$\hat{\omega}\in \Omega_\E^{n-1,2}(F)$ a \emph{presymplectic current
density compatible with $\E$} if it is both horizontally and vertically
closed,
\begin{align}
	\dh^\E \hat{\omega} &= 0 , \\
	\dv^\E \hat{\omega} &= 0 .
\end{align}
In other words, $\omega$ represents a cocycle (that is, a vertically
closed element) in the cohomology complex $(H_\E^{n-1,*}(\dh),\dv)$.
(Note that we are not bounding the jet order of $\hat{\omega}$.)
\end{definition}

The particular form $\iota^*\omega$ defined by Eq.~\eqref{eq:omega-def}
will be referred to as the presymplectic current density associated to
or obtained from the Lagrangian density $\L$, if there is any potential
confusion.

\section{Inverse problem}\label{sec:invprob}
In the preceding section we have defined variational PDE systems showed
that each one is endowed with a geometric structure (the presymplectic
current). The \emph{inverse problem of the calculus of variations} (or
the \emph{inverse problem} for short) is, given a bundle $F\to M$ and a
PDE system $\iota\colon \E\sso J^kF$ of order $k$, to decide when it is
variational.

The simpler \emph{non-multiplier} version of the inverse problem
presupposes that we are given an equation form $(f,\tilde{F}^*)$ for
$\E$. It consists of deciding whether there exists a Lagrangian density
whose Euler-Lagrange equations $(\EL,\tilde{F}^*)$ are \emph{equal} to
$(f,\tilde{F}^*)$. The necessary and sufficient conditions for the
non-multiplier inverse problem are known and are called the
\emph{Helmholtz conditions}. They can be formulated elegantly as the
requirement that the form $f_a\wedge \dv{u^a}\in \Forms^{n,1}(F)$ be
closed in a slightly extended version of the variational
bicomplex~\cite[Ch.5]{anderson-var}, \cite{sard}.

The harder \emph{multiplier} inverse problem consists of deciding
variationality directly from the sub-bundle $\iota\colon \E\sso J^kF$
itself or, equivalently, \emph{any} regular equation form $(f,E)$ of
$\E$. The name comes from the possibility of reducing it to the simpler
problem by finding the right set of ``multipliers'' $\epsilon$ (which
could also be differential operators) such that $(\epsilon\circ
f,\tilde{F}^*)$ satisfies the Helmholtz conditions. Unfortunately, the
multiplier inverse problem does not yet have a satisfactory solution in
full generality~\cite{anderson-aspects,saunders}.

An important contribution to the subject was made in~\cite{hen82}, as
discussed in the \hyperref[sec:intro]{Introduction}. Henneaux showed
that, for second order ODE systems that can be put into the canonical
form
\begin{equation}\label{eq:ode-form}
	\ddot{q} - f(t,q,\dot{q}) = 0 ,
\end{equation}
the existence of a symplectic form $\hat{\omega}(t,q,\dot{q})$, defined
on the bundle of initial data $(q,\dot{q})$ over the time axis, that is
conserved by the flow of the vector field associated to the ODE
system~\eqref{eq:ode-form},
\begin{equation}\label{eq:ode-symp}
	\del_t \hat{\omega} - \Lie_{f} \hat{\omega} = 0 ,
\end{equation}
is equivalent to this ODE system being variational with a unique
Lagrangian density $\L$ (up to addition of boundary terms) whose
associated symplectic form is equal to $\hat{\omega}$. Since the EL
equations of $\L$ will in general not be directly in the canonical
form~\eqref{eq:ode-form}, this result shows that a conserved symplectic
form $\hat{\omega}$ is a certificate of a positive solution of the
multiplier inverse problem. Henneaux's proof even provides a procedure
to construct $\L$ from $\hat{\omega}$ and the ODE system. The multiplier
inverse problem, is then reduced to identifying conserved symplectic
forms, which could be attacked by algebraic means.

Unfortunately, until rather recently, it has not been clear how to
generalize Henneaux's reformulation of the multiplier inverse problem to
PDEs~\cite{hen84}. Several aspects of the discussion in the previous paragraph are
specific to ODEs: (a) the possibility of a simple canonical form
like~\eqref{eq:ode-form}, (b) the geometric formulation of the ODE as a
vector field, (c) a local symplectic form $\hat{\omega}$, (d) the
conservation condition~\eqref{eq:ode-symp}. In this section, we present
a partial generalization of Henneaux's result to PDE systems. The PDE
analogs of the key aspects are (b) the geometric formulation in terms of
jet bundles (as in Section~\ref{sec:jets}), (c) the local covariant
presymplectic current density $\hat{\omega}$ (as in
Definition~\ref{def:presymp}), and (d) the closure condition
$\dh^\E\hat{\omega}=0$. Unfortunately, we have not been able to identify
simple, local analogs of the canonical form~\eqref{eq:ode-form} and
nondegeneracy of $\hat{\omega}$ (that is, being symplectic rather than
just presymplectic). Due to the last caveat, the procedure given below
does not produce a unique class of equivalent Lagrangian densities $\L$
associated to a given PDE system $\E$ and a presymplectic current
$\hat{\omega}$. On the other hand, each Lagrangian density produced is
in a certain sense a subsystem of $\E$: any solution of $\E$ also solves
the corresponding EL equations. Section~\ref{sec:uniq} is an attempt to
characterize the class of Lagrangian densities that can be so produced.

The key observation that connects a local presymplectic current density
with a variational formulation was made in~\cite[Sec.4]{bhl}, which is a
more geometric formulation of an earlier observation made in \cite{hydon}.
However, these authors did not attempt to place this result in the
context of other work on the inverse problem of the calculus of
variations and did not remark the similarity with the previous work of
Henneaux. Moreover, their calculations remained ``on-shell'', which
avoided lifting the Lagrangian density ``off-shell'' (see the proof
below), which is really necessary for a solution of the inverse problem.
Below, we clarify this observation and show in detail how a Lagrangian
density can be constructed using a method related to cohomological
descent~\cite{bbh} (see also Refs.~[88,89] and~[191] therein).

As before, consider a vector bundle $F\to M$ over an $n$-dimensional
manifold $M$ and a regular PDE system $\iota\colon\E\sso J^kF$ of order
$k$. Recall also that $\tilde{F}^* = \Lambda^n M\otimes_M F^*$ is the
densitized dual vector bundle of $F$. Finally, an important hypothesis
currently assumed is that the de~Rham cohomology groups of $\E^\oo$ and
$J^\oo F$ all vanish. All other relevant notions and notation are
defined in Section~\ref{sec:jets}.

\begin{theorem}
If there exists a presymplectic current density $\hat{\omega}\in
\Forms_\E^{n-1,2}(F)$ compatible with $\E$, then there exists a local
Lagrangian density $\L\in \Forms^{n,0}(F)$ such that the associated
presymplectic current density $\omega$ coincides with $\hat{\omega}$ on
solutions ($\iota^*_\oo\omega = \hat{\omega}$) and the
Euler-Lagrange PDE system $\E_{\EL}$, with equation form
$(\EL,\tilde{F}^*)$, is compatible with $\E$ ($\E^\oo \sse \E_{\EL}^\oo$
or all solutions of $\E$ also solve $\E_{\EL}$).
\end{theorem}
\begin{proof}
The fact that $\hat{\omega}$, from the definition of a presymplectic
current density, is both horizontally and vertically closed as an
element of $\Forms^{n-1,2}_\E(F)$ can be easily seen to be equivalent to
$\hat{\omega}$ being de~Rham closed as an element of
$\Forms^{n+1}(\E^\oo)$. Equivalently, if $\omega$ is of homogeneous
degrees $(n-1,2)$ and de~Rham closed, it follows that it is both
horizontally and vertically closed.

By assumption, the de~Rham cohomology group $H^{n+1}(\E^\oo)$ is trivial.
If it were not, there could be global topological obstructions to this
construction, which we do not discuss in the current treatment.  Then
there must exist an element $\hat{\rho}\in\Omega^n(\E^\oo)$ such that
$\d^\E\hat{\rho} = \hat{\omega}$. If we expand $\hat{\rho}$ in
components of homogeneous vertical and horizontal degrees, we can
represent it as $\hat{\rho} = \sum_{h+v=n} \hat{\rho}_{h,v}$. This sum
is finite, since $0\le h \le n$ and $0 \le v$. The equation
\begin{equation}
	\d^\E\hat{\rho} = \dv^\E\hat{\rho} + \dh^\E\hat{\rho} = \hat{\omega}
\end{equation}
then naturally expands into the following system for the homogeneous
components
\begin{align}
	\dv^\E \hat{\rho}_{0,n}
		&= 0 , \\
	\dv^\E \hat{\rho}_{1,n-1}
		&= -\dh^\E \hat{\rho}_{0,n} , \\
		&\vdots \\
	\dv^\E \hat{\rho}_{n-1,1}
		&= -\dh^\E \hat{\rho}_{n-2,2} + \hat{\omega} , \\
	\dv^\E \hat{\rho}_{n,0}
		&= -\dh^\E \hat{\rho}_{n-1,1} .
\end{align}
Note that $\hat{\rho}$ could be constructed by solving the above
equations term by term. This method is a special case of the
\emph{cohomological descent method}~\cite{bbh}.


As written, these equations hold ``on-shell,'' that is, on the
$\oo$-prolonged PDE manifold $\E^\oo$. But the descent equations lift
``off-shell,'' to the total space of jet bundle $J^\oo F$ containing
$\E^\oo$. Denote the lifted forms by removing hats, $\hat{\rho}_{h,v} =
\iota_\oo^*\rho_{h,v}$, with the exception $\hat{\omega} =
\iota_\oo^*\omega'$. These lifts are not unique, as we could always
change them by adding terms that are annihilated by the pullback to
$\E^\oo$. Recall that $\E^\oo$ is defined by the equations $p^\oo f=0$,
so given local coordinates $(x^i,v_{A})$ on the equation bundle $E$,
which extend to $(x^i,v_{IA})$ on $J^\oo E$, the terms annihilated by
the pullback to $\E^\oo$ must be proportional to $f_{IA}$ or the
exterior vertical derivatives $\dv f_{IA}$. After the lift, the above
equations for $\rho_{h,v}$ also only hold up to terms proportional to
$f_{IA}$ or $\dv f_{IA}$,
\begin{alignat*}{4}
	\dv \rho_{0,n}
		&= 0 &&{}
			+ f_{IA}\lambda^{IA}_{0,n+1} &&{}+ \dv f_{IA} \wedge \mu^{IA}_{0,n}, \\
	\dv \rho_{1,n-1}
		&= -\dh \rho_{0,n} &&{}
			+ f_{IA}\lambda^{IA}_{1,n} &&{}+ \dv f_{IA} \wedge \mu^{IA}_{1,n-1} , \\
		&~~\vdots \\
	\dv \rho_{n-2,2}
		&= -\dh \rho_{n-3,3} &&{}
			+ f_{IA}\lambda^{IA}_{n-2,3} &&{}+ \dv f_{IA} \wedge \mu^{IA}_{n-2,2} , \\
\intertext{and}
	\dv\rho_{n-1,1}
		&= -\dh\rho_{n-2,2} &&{}
			+ f_{IA}\lambda^{IA}_{n-1,2} &&{}+ \dv f_{IA} \wedge \mu^{IA}_{n-1,1} &{}
			+ \omega', \\
	\dv\rho_{n,0}
		&= -\dh\rho_{n-1,1} &&{}
			+ f_{IA}\lambda^{IA}_{n,1} &&{}+ \dv f_{IA} \wedge \mu^{IA}_{n,0} .
\end{alignat*}
Note that we have introduced the coefficient forms $\lambda^{IA}_{h,v}$
and $\mu^{IA}_{h,v}$ to parametrize the terms annihilated by the
pullback to $\E^\oo$; they are of homogeneous horizontal and vertical
degrees, as indicate by their subscripts. These coefficient forms are not
simply arbitrary. As shown below, $\lambda^{IA}_{n,1}$ and
$\mu^{IA}_{n,0}$ will contain in them the information about the
multipliers needed to solve the inverse problem.

Given local coordinates $(x^i,u^a_I)$ on $J^\oo F$, the goal now is to
construct the forms $\L$, $\theta$, $\omega$ and $\EL_a\wedge \dv u^a$,
of respective degrees $(n,0)$, $(n-1,1)$, $(n-1,2)$ and $(n,1)$, such
that the corresponding equations of the covariant phase space method
hold:
\begin{align}
\label{L-EL-th}
	\dv\L &= \EL_a \wedge \dv{u^a} - \dh\theta , \\
	\dv\theta &= \omega .
\end{align}
We rewrite the $(n,0)$-descent equation as
\begin{align}
	\dv(\rho_{n,0}-f_{IA}\mu^{IA}_{n,0})
	&= f_{IA}(\lambda^{IA}_{n,1}-\dv\mu^{IA}_{n,0}) - \dh\rho_{n-1,1} \\
	&= f_{IA} \epsilon^{IA}_a \wedge \dv{u^a}
		- \dh(\rho_{n-1,1} + f_{IA}\lambda^{\prime IA}_{n-1,1}) ,
\end{align}
where integration by parts was used to construct $\lambda^{\prime
IA}_{n-1,1}$ and $\epsilon^{IA}_a$, with the latter being $(n,0)$-forms.
We add $\dv(f_{IA} \lambda^{\prime IA}_{n-1,1})$ to both sides of the
$(n-1,1)$-descent equation and rewrite it as
\begin{multline}
	\dv(\rho_{n-1,1}+f_{IA}\lambda^{\prime IA}_{n-1,1})
	= \omega' -\dh\rho_{n-2,2} + f_{IA} (\lambda^{IA}_{n-1,2} + \dv\lambda^{\prime IA}_{n-1,1}) \\
		{} + \dv{f_{IA}}\wedge(\mu^{IA}_{n-1,1}+\lambda^{\prime IA}_{n-1,1}) .
\end{multline}
It is now clear that these equations take the desired form with the
following identifications:
\begin{align}
	\L &= \rho_{n,0} - f_{IA} \mu^{IA}_{n,0} , \\
	\theta &= \rho_{n-1,1} + f_{IA} \lambda^{\prime IA}_{n-1,1} , \\
	\omega &= \omega' -\dh\rho_{n-2,2} \\
\notag & \qquad {}
		+ f_{IA} (\lambda^{IA}_{n-1,2} + \dv\lambda^{\prime IA}_{n-1,1})
		+ \dv{f_{IA}}\wedge(\mu^{IA}_{n-1,1}+\lambda^{\prime IA}_{n-1,1}) , \\
	\EL_a\wedge \dv{u^a} &= f_{IA} \epsilon^{IA}_a \wedge \dv{u^a} .
\end{align}
The form $\EL_a\wedge \dv{u^a}$ naturally correspond to a bundle
morphism $\EL\colon J^l F\to \tilde{F}^*$, for some finite jet degree
$l$, which defines the Euler-Lagrange PDE system $\E_{\EL}\sso J^lF$ via
the equation form $(\EL,\tilde{F}^*)$. From the last equation, it is
obvious that the constructed Euler-Lagrange PDE system contains the
original one, $\E^\oo \sse \E_{\EL}^\oo$.
\end{proof}

Note that all the above calculations were done on jet bundles of
infinite order. However, each of the forms introduced at intermediate
steps depends only on jet coordinates of some finite order. Therefore,
the final Lagrangian will also depend on jet coordinates up to some
finite order, which may be much higher than the order of the original
PDE system. Note that the degree of the Lagrangian could be artificially
inflated by the presence of boundary terms like $\dh\B$, where $\B$
could depend on jet coordinates of some high order.

It is also important to remark that the resulting Euler-Lagrange
equations may not be equivalent to the full original PDE system, but
only to a subsystem thereof or a ``weaker'' system, one whose solution
space contains all solutions of $\E$, but may be strictly larger. This
is unavoidable, since the PDE system may consist, for example, of
several uncoupled subsystems, each of which may have an independent
variational formulation. More complicated situations are of course
possible. In the ODE context, the requirements that $\hat{\omega}$
actually be symplectic (rather than just presymplectic) and that the ODE
system is in the canonical form~\eqref{eq:ode-form} are sufficient to
guarantee that the original system is fully variational~\cite{hen82}.
Unfortunately, at least one of these conditions fails already when the
ODE system is not determined (possibly because of gauge invariance) or
under the inclusion of algebraic equations that constrain the initial
data $(q,\dot{q})$.  In the PDE case, in analogy with the ODE one, there
may be a set of conditions on the symbol of the PDE system and on
$\hat{\omega}$ that guarantees a fully variational formulation, which
may be further complicated by allowing equations with gauge symmetries
or constraints. However, such analogous conditions remain to be
investigated in detail.

\section{Arbitrariness in Lagrangian density construction}\label{sec:uniq}
There were a number of choices involved in the construction of the
(off-shell) Lagrangian density $\L$ from the (on-shell) presymplectic
current density $\hat{\omega}$. In this section, we investigate how the
resulting $\L$ depends on these choices.

The choices are exhausted by the following substitutions:
\begin{align}
	\text{(i)} ~~ \hat{\omega} &\to \hat{\omega}
		+ \dh^\E \hat{\pi}_{n-2,2} ,
	\quad \text{with} ~~
		\dv^\E \hat{\pi}_{n-2,2} = 0 ; \\
	\text{(ii)} ~~ \hat{\rho} &\to \hat{\rho} + \d^\E \hat{\sigma} ; \\
	\text{(iii)} ~~ \rho &\to \rho
		+ f_{IA} \bar{\lambda}^{IA} + \dv f_{IA} \wedge \bar{\mu}^{IA} , \\
		\omega' &\to \omega' + f_{IA} \bar{\lambda}^{\prime IA}
			+ \dv f_{IA}\wedge \bar{\mu}^{\prime IA} .
\end{align}
We deal with each kind of substitution one by one.

\textit{(i)} It is easy to see the following identity:
\begin{equation}
	\d^\E(\hat{\rho} + \hat{\pi}_{n-2,2})
		= \hat{\omega} + \dh^\E\hat{\pi}_{n-2,2} .
\end{equation}
Since the change $\hat{\rho}\to \hat{\rho} + \hat{\pi}_{n-2,2}$ affects
neither of the $\hat{\rho}_{n,0}$ or $\hat{\rho}_{n-1,1}$ components,
the off-shell lift of the $(n,0)$ descent equation is unmodified.
Therefore, the Lagrangian density $\L$ does not change.

\textit{(ii)} Lifting $\hat{\sigma}$ off-shell to $\sigma$ and using
the decomposition into homogeneous components, $\sigma = \sum_{h+v=n}
\sigma_{h,v}$, we see the change \begin{equation}
	\rho_{h,v} \to \rho_{h,v} + \dh\sigma_{h-1,v} +
		\dv\sigma_{h,v-1}.
\end{equation}
However, since $\d(\rho + \d\sigma) = \d\rho$, the descent equations are
unmodified. In the end, the Lagrangian changes only as $\L \to \L +
\dh\sigma_{n-1,0}$. Since the change is by a horizontally exact term,
the EL equations remain the same.

\textit{(iii)} This kind of substitution does in general change the
equivalence class of the Lagrangian density, that is, the EL equations
of the modified Lagrangian density may be different. The forms
parametrizing the failure of the off-shell lift of the descent equations
undergo the change (except for the $\lambda_{n-1,2}$ and $\mu_{n-1,1}$
coefficients, which undergo obvious additional changes compensating
the change in $\omega'$, which ultimately do not affect $\L$)
\begin{align}
	\lambda_{h,v}^{iIA} &\to \lambda_{h,v}^{iIA}
		+ (\dh\bar{\lambda}^{iIA}_{h-1,v}
			+\d{x^{(i}}\wedge\bar{\lambda}^{I)A}_{h-1,v})
		+ \dv \bar{\lambda}^{iIA}_{h,v-1} , \\
	\mu_{h,v}^{iIA} &\to \mu_{h,v}^{iIA}
		- (\dh\bar{\mu}^{iIA}_{h-1,v}+\d{x^{(i}}\wedge \bar{\mu}^{I)A}_{h-1,v})
		+ (\bar{\lambda}^{iIA}_{h,v}-\dv\bar{\mu}^{iIA}_{h,v-1}) .
\end{align}
Therefore, the Lagrangian density undergoes the change
\begin{align}
	\L \to \L+\bar{\L} &= \L + f_{iIA} \bar{\lambda}^{iIA}_{n,0}
		- f_{iIA} [ \bar{\lambda}^{iIA}_{n,0}
			-(\dh\bar{\mu}^{iIA}_{n-1,0} + \d{x^{(i}}\wedge\bar{\mu}^{I)A}_{n-1,0}) ] \\
		&= \L + f_{iIA}(\dh\bar{\mu}^{iIA}_{n-1,0}
			+ \d{x^{(i}}\wedge\bar{\mu}^{I)A}_{n-1,0}) .
\end{align}
In general, the multipliers $\epsilon_a^{IA}$ will change as well, say
$\epsilon_a^{IA} \to \epsilon_a^{IA} + \bar{\epsilon}_a^{IA}$. Thus, the
EL equations of $\L$ and $\L+\bar{\L}$ may not be equivalent. In other
words, the equivalence classes $[\L]$ and $[\L+\bar{\L}]$ will differ.
Though both sets of EL equations will be consequences of $\E$.

At this point, it is worth reflecting on when two Lagrangians $\L$ and
$\L'$ should be considered equivalent. The standard answer is iff they
differ by a boundary term, $\L'-\L = \dh\B$. However, consider the
simple $1$-dimensional Lagrangians
\begin{align}
	\L[q_1,q_2,\lambda]
		&= \left[\frac{1}{2}\dot{q}_1^2 + \frac{1}{2}\dot{q}_2^2
			+ \lambda(q_1-q_2) \right] \d{t} , \\
	\L'[q_1,q_2,\lambda]
		&= \left[\frac{1}{2}\dot{q}_1^2 + \frac{1}{2}\dot{q}_2^2
			+ (\lambda + \alpha)(q_1-q_2) \right] \d{t} ,
\end{align}
where $\alpha$ is a constant (though in principle it could be a more
complicated function of $\lambda$ and $q_i$). Note that the difference
between the two Lagrangians, $\L'-\L = \alpha(q_1-q_2)$, is not a
boundary term, though it is proportional to the constraint equation
$q_2-q_1=0$ obtained by varying $\lambda$. It is easy to check that the
resulting Euler-Lagrange equations for either Lagrangian are equivalent
to the set $\ddot{q}_1=0$, $q_2=q_1$, $\lambda=0$. Moreover, their
symplectic currents agree as well, $\omega = \sum_i\dv\dot{q}_i\wedge\dv
q_i$ = $\omega'$. Furthermore, since we never limited the order of the
Lagrangian density, it could be of order zero (a ordinary,
non-differential variational problem). Then both $\omega=\omega'=0$ and
the critical points of $\L$ and $\L'$ coincide as long as $\dv(\L'-\L)$
vanishes at all critical points. Therefore, in the presence of
constraints on the initial data, the classes of equivalent Lagrangian
densities are larger than just those that differ by boundary terms, if
equivalence is evaluated in terms of the Euler-Lagrange equations and
the presymplectic current.

Consider the following preorder%
	\footnote{A \emph{preorder} $P$ on a set $X$ is a relation such that
	is \emph{reflexive} ($xPx$) and \emph{transitive} ($xPy$ and $yPz$ implies
	$xPz$), but in general neither symmetric nor antisymmetric. All
	partial orders and equivalence relations are preorders. The maximum
	symmetric subrelation $E$ ($xEy$ iff $xPy$ and $yPx$) is an
	equivalence relation and the quotient $X\to X/E$ projects $P$ to a
	partial order $P/E$ on $X/E$ (which is necessarily
	antisymmetric)~\cite{schr-ord}.} %
on Lagrangian densities. We say that $\L'$ contains $\L$, $\L \prec
\L'$, if the $\oo$-jet bundle submanifolds defined by the prolongations
of the corresponding Euler-Lagrange equations obey $\E^\oo_{\EL} \sse
\E^\oo_{\EL'}$ and their presymplectic current densities agree up to
a horizontally exact term when pulled back to the more restrictive
equation, $\iota_\oo^*\omega = \iota_\oo^*(\omega' + \dh\pi)$, where $\iota\colon
\E_{\EL} \sso J^\oo F$. The relation $\prec$ is clearly transitive and
reflexive and hence a preorder.

The containment relation is however not a partial order, because it is
not antisymmetric. As already discussed earlier, if $\L'-\L = \dh\B$,
then $\L\prec \L'$ and $\L'\prec \L$, while $\L\ne \L'$. On the other
hand, this preorder gives rise to the equivalence relation $\L \sim \L'$
(iff $\L\prec \L'$ and $\L'\prec \L$) and projects to an actual partial
order $[\L]\prec [\L']$ on the equivalence classes with respect to
$\sim$.

Given a PDE system $\iota\colon \E\sso J^k F$ with a compatible
presymplectic current density $\hat{\omega}$, define the set
$L_{\E,\hat{\omega}}$ to consist of equivalence classes of Lagrangian
densities $[\L]$ such that $\E^\oo \sse \E^\oo_{\EL}$ and
$\iota_\oo^*\omega = \iota_\oo^*(\hat{\omega} + \dh\pi)$. Then $L_{\E,\hat{\omega}}$ is
an upper set with respect to the $\prec$ partial order (which means that
$[\L]\in L_{\E,\hat{\omega}}$ and $[\L]\prec [\L']$ implies $[\L']\in
L_{\E,\hat{\omega}}$). After all, as should be clear from the definition
of the $\prec$ relation, if $\rho = \L + \theta$ could arise in the
construction in the preceding section (that is $\L\in
L_{\E,\hat{\omega}}$), then so could $\rho = \L' + \theta'$ (that is
$\L'\in L_{\E,\hat{\omega}}$) for any $\L'$ that contains $\L$ ($\L\prec
\L'$).

It is clear that the construction of the preceding section can yield any
element of $L_{\E,\hat{\omega}}$ as a partial solution of the inverse
problem. It is possible that a consideration of the structure of the
partial order $\prec$ as well as the decomposability and minimal
elements of $L_{\E,\hat{\omega}}$ in relation with the symbol of the PDE
system $\E^\oo$ can yield a definite solution of the inverse problem.
These questions are yet to be investigated in detail.

\section{Conclusion}\label{sec:concl}
We have shown, following a recent observation in~\cite{bhl,hydon} and a
strong analogy with previous work in~\cite{hen82,juras,at}, that a horizontally
conserved presymplectic current density is a certificate that a
subsystem of a given PDE system or a comparatively weaker system is
variational (any solution of the original PDE system also solves the
obtained variational system). If this subsystem is actually the full
system, then the multiplier inverse problem of the calculus of
variations has a positive solution.

Restricting to the context of second order ODEs that can be put in
canonical form~\eqref{eq:ode-form}, it is known that any two Lagrangians
whose Euler-Lagrange equations are equivalent to the given ODE system
$\E$ and that have equivalent symplectic current densities
$\hat{\omega}$ must differ by a boundary term \cite{hen82}. In that case,
the equivalence class of Lagrangians solving the full inverse problem of
the calculus of variations for $\E$ depends only on the characteristic
cohomology class in $[\hat{\omega}]\in H^{n-1,2}_\E(\dh)$. However, in
more general contexts (for general PDEs, or even ODEs with constraints)
a characterization of equivalent Lagrangians (those sharing the same set
of solutions and equivalent presymplectic current densities) is still
missing.

On the other hand, in Section~\ref{sec:uniq} we defined a preorder
relation $\prec$ on Lagrangian densities $\L\in\Forms^{n,0}(F)$ that
relates both this equivalence problem for Lagrangian densities and the
ambiguity in the Lagrangians produced by the construction of
Section~\ref{sec:invprob}. The precise structure of the preorder $\prec$
and the conditions on a PDE system and compatible presymplectic current
that guarantee that it is fully variational (rather than just a
subsystem thereof) remain to be investigated. It is possible that
generic methods for determining the characteristic cohomology groups of
a PDE~\cite{bg,kv} would be helpful in classifying possible
presymplectic current densities and hence variational formulations.

\section*{Acknowledgments}
The author would like to thank Urs Schreiber for helpful discussions and
Chris Rogers for pointing out Ref.~\cite{bhl} during a visit to Utrecht
in October, 2011. Also, thanks to Glenn Barnich, Peter Olver and Ian
Anderson for their encouragement.

\bibliographystyle{utphys}
\bibliography{paper-invprob}

\end{document}